\documentclass[journal]{IEEEtran}

\usepackage{amsmath}
\usepackage{amsfonts}
\usepackage{amssymb}
\usepackage{amsthm}
\usepackage{tabularx}
\usepackage{mathrsfs} 
\usepackage{algorithmic}
\usepackage{algorithm} 
\usepackage{url}
\usepackage{hyperref}

\newtheorem{definition}{Definition}
\newtheorem{Proposition}{Proposition}
\newtheorem{lemma}{Lemma}
\newtheorem{corollary}{Corollary}

\usepackage{stfloats}
\usepackage{float}
\usepackage{graphicx}
\usepackage{cite}
\usepackage{xcolor}
\usepackage{subfigure}
\renewcommand{\vec}[1]{\mathbf{#1}}

\makeatletter
\def\blfootnote{\xdef\@thefnmark{}\@footnotetext}
\makeatother

\begin{document}
\title{\LARGE Fluid Antenna-Assisted Dirty Multiple Access Channels\\over Composite Fading\thanks{The work of K. Wong and K. Tong is supported by the Engineering and Physical Sciences Research Council (EPSRC) under Grant EP/W026813/1. For the purpose of open access, the authors will apply a Creative Commons Attribution (CC BY) licence to any Author Accepted Manuscript version arising. The work of L\'opez-Mart\'inez was funded in part by Consejer\'ia de Transformaci\'on Econ\'omica, Industria, Conocimiento y Universidades of Junta de Andaluc\'ia, and in part by MCIN/AEI/10.13039/501100011033 through grants EMERGIA20 00297 and PID2020-118139RB-I00. The work of C. B. Chae is supported by the Institute of Information and Communication Technology Promotion (IITP) grant funded by the Ministry of Science and ICT (MSIT), Korea (No. 2021-0-02208, No. 2021-0-00486).}
\thanks{F. R. Ghadi and F. J. L\'opez-Mart\'inez are with the Communications and Signal Processing Lab, Telecommunication Research Institute (TELMA), Universidad de M\'alaga, M\'alaga, 29010, Spain. F. J. L\'opez-Mart\'inez is also with the Department of Signal Theory, Networking and Communications, University of Granada, 18071, Granada, Spain (e-mail: $\rm farshad@ic.uma.es$, $\rm fjlm@ugr.es$).}
\thanks{K. Wong and K. Tong are with the Department of Electronic and Electrical Engineering, University College London, London WC1E 6BT, United Kingdom. K. Wong is also with Yonsei Frontier Lab, Yonsei University, Seoul, 03722, Korea (e-mail: $\{\rm kai\text{-}kit.wong,k.tong\}@ucl.ac.uk$).}
\thanks{C. B. Chae is with School of Integrated Technology, Yonsei University, Seoul, 03722, Korea (e-mail: $\rm cbchae@yonsei.ac.kr$).}
\thanks{Y. Zhang is with Kuang-Chi Science Limited, Hong Kong SAR, China (e-mail: $\rm yangyang.zhang@kuang\text{-}chi.com$).}
\thanks{Corresponding author: Kai-Kit Wong.}
}
  
\author{Farshad~Rostami~Ghadi\IEEEmembership{}, 
            Kai-Kit~Wong, \IEEEmembership{Fellow}, \textit{IEEE},\\
  	    F. Javier~L\'opez-Mart\'inez, \IEEEmembership{Senior Member}, \textit{IEEE}, 
	    Chan-Byoung Chae, \IEEEmembership{Fellow}, \textit{IEEE},\\
	    Kin-Fai Tong, \IEEEmembership{Fellow}, \textit{IEEE}, and Yangyang Zhang

\vspace{-7mm}
}

\maketitle

\begin{abstract}
This letter investigates the application of the emerging fluid antenna (FA) technology in multiuser communication systems when side information (SI) is available at the transmitters. In particular, we consider a $K$-user dirty multiple access channel (DMAC) with non-causally known SI at the transmitters, where $K$ users send independent messages to a common receiver with a FA capable of changing its location depending on the channel condition. By connecting Jakes' model to copula theory through Spearman's $\rho$ rank correlation coefficient, we accurately describe the spatial correlation between the FA channels, and derive a closed-form expression for the outage probability (OP) under Fisher-Snedecor $\mathcal{F}$ fading. Numerical results illustrate how considering FA can improve the performance of multiuser communication systems in terms of the OP and also support a large number of users using only one FA at the common receiver in a few wavelengths of space.
\end{abstract}


\maketitle
\blfootnote{Digital Object Identifier 10.1109/XXX.2021.XXXXXXX}
\vspace{0mm}

\section{Introduction}\label{sec-intro}  
Multiple-input multiple-output (MIMO) systems, including multiuser MIMO and massive MIMO, have become one of the greatest advances in mobile communications over recent years. In this regard, deploying multiple antennas far apart at fixed locations at the transmitter and/or receiver sides enables MIMO to multiplex users entirely in the spatial domain and enhances the network capacity, e.g., \cite{larsson2014massive}. In the fifth-generation (5G) of wireless communications, massive MIMO can efficiently simplify the signal processing of multiuser signals by adopting a massive number of antennas (e.g., over $64$) at the base stations (BSs) \cite{ngo2013energy}. The sixth-generation (6G) will also see extra-large MIMO succeed by even utilizing continuous-aperture MIMO (CP-MIMO) \cite{huang2020holographic}. However, despite the outstanding advantages of different types of MIMO systems, there are practical issues such as the complexity of precoding optimization and channel state information (CSI) acquisition. In addition, ensuring the minimum distance between antenna spacing to be half of the wavelength can limit the integration of MIMO systems inside mobile devices due to physical space limitations \cite{stuber2001principles}.  
	
In order to address the aforesaid challenges in future mobile communication technology, fluid antenna (FA) systems have recently been introduced in \cite{wong2020fluid,wong2020fluid1} to obtain additional diversity and multiplexing benefits by exploiting novel dynamic radiating structures. In contrast to conventional antenna systems, a FA is a software-controllable fluidic (or pixel-based), conductive, or dielectric structure \cite{huang2021liquid} that has the ability to switch its location (i.e., ports) in a given small space. Therefore, these unique features can provide a stronger channel gain, lower interference, and other desirable performance for future wireless communications. In this regard, great efforts have recently been carried out to evaluate the FA system. Initial performance analysis for a single-user FA system was studied in \cite{wong2020fluid} and \cite{wong2020performance}, where the authors derived analytical expressions of the outage probability (OP) and ergodic capacity over correlated Rayleigh fading channels, respectively. Additionally, an analytical expression of the OP for a point-to-point (P2P) single FA system under correlated Nakagami-$m$ fading channels was obtained in \cite{tlebaldiyeva2022enhancing}.  Moreover, the closed-form expression of the spatial correlation parameters for a P2P FA system was obtained in \cite{wong2022closed} and a general eigenvalue-based model to accurately approximate the spatial correlation between FA ports given by Jakes' model was proposed in \cite{khammassi2023new}. In addition, only recently, a general copula-based formulation was proposed in \cite{ghadi2023copula} for the joint multivariate distribution of arbitrary correlated fading channels in the FA system which can describe the spatial correlation between FA ports beyond linear dependence structures. Besides, apart from the above-mentioned single P2P FA systems, the authors in \cite{wong2021fluid} evaluated the performance of fluid antenna multiple access (FAMA) in terms of the OP and average capacity.  

Motivated by the great potential of FA and the importance of analyzing the FA system in realistic radio propagation environments, this letter aims to investigate the performance of FA-aided multiuser communication systems under generalized fading conditions. Specifically, we consider a $K$-user dirty multiple access channel (DMAC) in which each user takes advantage of non-causally available side information (SI) and wants to send an independent message to a common BS receiver equipped with an FA system through correlated Fisher-Snedecor $\mathcal{F}$ fading channels. In this model, SI is an information-theoretic concept that mainly refers to either CSI or interference awareness, and has the ability to meet the reliability constraint in future wireless networks by reducing the destructive effects of the interference and providing reliable communication at higher rates \cite{philosof2011lattice,ghadi2023ris,ghadi2021role}. In particular, by connecting Jakes' model to copula theory with the help of rank correlation coefficients such as Spearman's $\rho$, we first provide an analytical expression for the joint cumulative distribution function (CDF) in the considered multiuser FA system under correlated Fisher-Snedecor $\mathcal{F}$ fading channels, and then derive the OP in a closed-form expression. It is worth noting that copula theory is a flexible statistical approach that can accurately generate the multivariate distributions of two or more arbitrarily correlated random variables and describe the corresponding dependence structure beyond linear correlation. This has recently become quite popular in the performance analysis of wireless communication systems \cite{ghadi2020copula,besser2020copula,ghadi2022performance}. Eventually, our numerical results illustrate how considering the FA in multiuser communication systems can significantly enhance the system performance in terms of the OP by supporting a large number of users. 

\section{System Model}\label{sec-sys}
We consider a $K$-user wireless multiple access communication system with known interference $S_k$, $k\in\{1,\dots,K\}$, where $K$ single fixed-antenna transmitters $t_k$ aim to send independent messages $X_k$ to a common receiver $r$ equipped with a single FA, respectively. We assume that the interference signals $S_k$ with zero mean and variances $Q_k$ (i.e., $S_k\sim\mathcal{N}\left(0,Q_k\right)$) are known non-causally to the transmitters $t_k$, and the information signals $X_k$ sent by transmitters $t_k$ are subjected to the average power constraint $\mathbb{E}\left[|X_k|^2\right]\leq P$, respectively. Moreover, we suppose that the FA includes only one radio frequency (RF) chain and $N$ preset positions (i.e., ports), which are evenly distributed on a linear space of length $W\lambda$ where $\lambda$ denotes the wavelength of the radiation. Thus, relative to the first port, the distance between the first port and the $n$-th port is given by  
\begin{align}
d_n=\left(\frac{n-1}{N-1}\right)W\lambda, \quad n=1,2,\dots,N.
\end{align}

Furthermore, the received signal at the $n$-th port of the common receiver $r$ can be written as
\begin{align}
Y_n=\sum_{k=1}^{K}h_{n,k}X_k+S_k+Z_n
\end{align}
in which $h_{n,k}$ denotes the fading channel coefficient of the $k$-th port from the $n$-th transmitter and $Z_n$ is the independent identically distributed (i.i.d.) additive white Gaussian noise (AWGN) with zero mean and variance $\sigma^2$ at each port. Under this model, we assume for the sake of generality that the square of the channel amplitude $|h_{n,k}|^2$ follows Fisher-Snedecor $\mathcal{F}$ distribution with the following CDF \cite{ghadi2022performance}
\begin{align}\label{eq-cdf1}
F_{|h_{n,k}|^2}(v)=I_{\frac{m_1v}{m_1v+m_2}}\left(\frac{m_1}{2},\frac{m_2}{2}\right),
\end{align}
where $I_x(a,b)$ is the regularized incomplete beta function and ($m_1,m_2$) are the corresponding degrees of freedom. 
In addition, we assume that the FA can always select the best port with the strongest signal for communication, i.e., 
\begin{align}
h_{\mathrm{FAS},k}=\max\left\{|h_{1,k}|,|h_{2,k}|,\dots,|h_{N,k}|\right\},
\end{align}
where the channel coefficients $h_{n,k}$ for $n\in\{1,\dots,N\}$ at the selected transmitter $t_k$ are spatially correlated since they can be arbitrarily close to each other. By assuming 2-D isotropic scattering and isotropic receiver ports on the FA, such spatial correlation can be characterized by Jakes' model as \cite{stuber2001principles}
\begin{align}\label{eq-Jake}
	\eta_{n}^{(k)}=\delta^2J_0\left(\frac{2\pi\left(n-1\right)}{N-1}W\right),
\end{align}
where $\eta_{n}^{(k)}$ denotes the linear correlation coefficient that can describe the correlation between the corresponding 
 channel coefficients for each transmitter $k$, $\delta^2$ accounts for the large-scale fading effect, and $J_0(\cdot)$ denotes the zero-order Bessel function of the first kind. Therefore, under such assumptions, the received signal-to-noise ratio (SNR) at the common FA receiver from the transmitter $t_k$ can be defined as
\begin{align}
\gamma_{k}=\frac{Ph_{\mathrm{FAS},k}^2}{\sigma^2}=\bar{\gamma}h_{\mathrm{FAS},k}^2,
\end{align}
where $\bar{\gamma}=\frac{P}{\sigma^2}$ represents the average SNR.

Further, in a $K$-user block fading  multiple access communications with side non-causally known SI at the transmitters $t_k$ and the coherent FA receiver $r$ (i.e., known as DMAC), the instantaneous capacity region is given by \cite{philosof2011lattice}
\begin{align}\label{eq-region}
\sum_{k=1}^{K}R_k\leq\frac{1}{2}\log_2\left(1+\underset{k=1,\dots,K}{\min}\frac{Ph_{\mathrm{FAS},k}^2}{\sigma^2}\right).
\end{align}

\section{Performance Analysis}
Here, we first describe the spatial correlation between the FA ports by combining Jakes' model to copula theory. Then, we derive a closed-form expression of the OP for the considered system under Fisher-Snedecor $\mathcal{F}$ fading channels.

\subsection{Statistical Characterization}
\begin{definition}[$d$-dimensional copula]
Let $\vec{S}=(S_1,\dots,S_d)$ be a vector of $d$ random variables (RVs) with marginal and joint CDFs $F_{S_i}(s_i)$ and $F_{S_1,\dots,S_d}(s_1,\dots,s_d)$ for $i\in\{1,\dots,d\}$, respectively. Then the copula function $C(u_1,\dots,u_d)$ of the random vector $\vec{S}$ defined on the unit hypercube $[0,1]^d$ with uniformly distributed RVs $U_i:=F_{S_i}(s_i)$ over $[0,1]$ is given by \cite{nelsen2006introduction}
\begin{equation}
C(u_1,\dots,u_d)=\Pr(U_1\leq u_1,\dots,U_d\leq u_d),
\end{equation}
where $u_i=F_{S_i}(s_i)$.
\end{definition}

Under such an assumption, Sklar's theorem postulates that there exists one Copula function $C$ such that for all $s_i$ in the extended real line domain $\bar{R}$ \cite{nelsen2006introduction}
\begin{equation}\label{eq-sklar}
F_{S_1,\dots,S_d}(s_1,\dots,s_d)=C\left(F_{S_1}(s_1),\dots,F_{S_d}(s_d)\right).
\end{equation}
It should be noted that the CDF provided by \eqref{eq-sklar} is valid for any choice of arbitrary correlated RVs from different families of marginal distributions, where by applying an appropriate copula function $C$ to it, the corresponding multivariate joint distribution is derived. However, selecting the optimum copula function that perfectly matches to the considered problem is a very challenging task.

In general, there are three main types of copulas that can be used to describe the linear/non-linear correlation between RVs, namely: empirical, elliptical, and Archimedean copulas. The empirical class of copula requires a set of existing real data for analyzing the structure of dependency \cite{ruschendorf2009distributional}. Elliptical copulas are defined as the dependence structure of some related elliptical distribution and can be obtained from the respective multivariate distribution function by standardizing the univariate marginal laws, which have complicated form \cite{frahm2003elliptical}. On the other hand, Archimedean copulas are normally used to correlate a potentially large number of similar RVs by a generator function \cite{nelsen1997dependence}, which have closed-form expressions. Archimedean copulas have also simple structures with nice properties and they are easy to construct so that the great variety of families of copulas belong to this class \cite{nelsen2006introduction}. In this regard, the empirical data analysis shows that the Clayton copula is one of the most popular Archimedean copulas which can perfectly describe any specific non-zero level of low tail dependency between individual RVs. Hence, since the outage of transmission often occurs in deep fading conditions (i.e., tail dependence) over wireless communication systems \cite{livieratos2014correlated}, we here exploit the Clayton copula to describe the structure of dependency between correlated fading channel coefficients and generate the corresponding joint CDF. Our analytical and numerical results will reveal that this choice is justified since it can offer good mathematical tractability and accurately describe the impact of spatial correlation between the FA ports in the considered multiuser communication setup. 

\begin{definition}[Clayton copula]
The $d$-dimension Clayton copula with a generator function $\phi(t)=\frac{t^{-\beta}-1}{\beta}$ is defined as
\begin{equation}\label{eq-cl}
C_{\mathrm{CL}}(u_1,u_2,\dots,u_d)=\left[\sum_{j=1}^d \left(u_j^{-\beta}-1\right)+1\right]^{-\frac{1}{\beta}},
\end{equation}
where $\beta\in[0,\infty)$ is a dependence structure parameter of Clayton copula. The independent case is achieved if $\beta=0$.
\end{definition}

Here, by utilizing the Clayton copula definition, we derive the CDF of $h_{\mathrm{FAS},k}^2$ in the following lemma.

\begin{lemma}
The CDF of $h_{\mathrm{FAS},k}^2$ under Fisher-Snedecor $\mathcal{F}$ fading channels is given by
\begin{equation}\label{eq-cdf2}
	F_{h_{\mathrm{FAS},k}^2}(r)=\left[\sum_{n=1}^N \left(I^{-\beta}_{\frac{m_1r}{m_1r+m_2}}\left(\frac{m_1}{2},\frac{m_2}{2}\right)-1\right)+1\right]^{-\frac{1}{\beta}}.
\end{equation}
\end{lemma}

\begin{proof}
By exploiting the definition of the CDF, $F_{h_{{\mathrm{FAS}},k}^2}(r)$ can be mathematically defined as
\begin{align}
F_{h_{{\mathrm{FAS}},k}^2}(r)&\hspace{.5mm}=\Pr\left(h_{{\mathrm{FAS}},k}^2\leq r\right)\notag\\
&\hspace{.5mm}=\Pr\left(\max\left\{|h_{1,k}|^2,\dots,|h_{N,k}|^2\right\}\leq r\right)\notag\\
&\hspace{.5mm}=F_{|h_{1,k}|^2,\dots,|h_{N,k}|^2}\left(r,\dots,r\right)\notag\\
&\overset{(a)}{=}C\left(F_{|h_{1,k}|^2}(r),\dots,F_{|h_{N,k}|^2}(r)\right),\label{eq-p1}
\end{align}
where ($a$) is obtained from \eqref{eq-sklar}. Now, by inserting \eqref{eq-cdf1} into  \eqref{eq-p1} and considering \eqref{eq-cl}, the proof is completed. 
\end{proof}

Note from \eqref{eq-cdf2} that the Clayton dependence parameter $\beta$ can generally capture the positive correlation between the FA ports as well as the independent case. However, it does not necessarily describe the linear correlation because when non-linear transformations are applied to those RVs, the linear correlation cannot be maintained anymore, i.e., copula definition. Additionally, it can be seen that the CDF of $h_{\mathrm{FAS},k}^2$ is only expressed in terms of the number of FA ports $N$ but it is widely accepted that the FA size $W$ has more effects on the spatial correlation between the fading coefficients \cite{wong2022closed}. In this regard, we need to measure such dependency by connecting the Clayton copula parameter $\beta$ to the Jakes' model. To this end, we exploit popular rank correlation coefficients such as Spearman's $\rho$, denoted by $\rho_\mathrm{s}$, that can measure the statistical dependence between the ranking of two RVs as follows
\begin{align}\label{eq-def-rho-s}
\rho_\mathrm{s}=12\iint_{\left[0,1\right]^2}u_1u_2dC\left(u_1,u_2\right)-3.
\end{align}
Now, by substituting the Clayton copula from \eqref{eq-cl} into \eqref{eq-def-rho-s} and computing the integral, $\rho_\mathrm{s}$ under Clayton copula can be approximated as
\begin{align}
\rho_\mathrm{s}\approx\frac{3\beta}{2\left(\beta+2\right)}.
\end{align}
Further, regarding the fact that the Spearman's $\rho$, i.e., $\rho_\mathrm{s}$, for a pair of continuous RVs is identical to Pearson's product moment correlation coefficient for the grades of the corresponding RVs (i.e., $\rho_\mathrm{s}=\eta_{n}^{(k)}$) \cite{nelsen2006introduction}, we can express the Clayton copula parameter $\beta$ in terms of Jakes' model as
\begin{align}\label{eq-beta}
\beta=\frac{4\eta_{n}^{(k)}}{3-2\eta_{n}^{(k)}}.
\end{align}
Thus, the CDF of $h_{\mathrm{FAS},k}^2$ can be determined in terms of the size $W$ and the number of FA ports $N$ using the following corollary. 

\begin{corollary}
By plugging $\beta$ from \eqref{eq-beta} into \eqref{eq-cdf2}, the CDF of $h_{\mathrm{FAS},k}^2$ under Fisher-Snedecor $\mathcal{F}$ fading channels with the correlation coefficient $\eta_{n}^{(k)}$ is given by 
\begin{equation}\label{eq-cdf3}
	F_{h_{\mathrm{FAS},k}^2}(r)=\left[\sum_{n=1}^N \left(I_{\frac{m_1r}{m_1r+m_2}}^{\frac{4\eta_{n}^{(k)}}{2\eta_{n}^{(k)}-3}}\left(\frac{m_1}{2},\frac{m_2}{2}\right)-1\right)+1\right]^{\frac{2\eta_{n}^{(k)}-3}{4\eta_{n}^{(k)}}},
\end{equation}
where $\eta_{n}^{(k)}$ is obtained from \eqref{eq-Jake}.
\end{corollary}

\subsection{Outage Probability}
\begin{Proposition}
The OP for the considered dirty FA multiuser communication system under Fisher-Snedecor $\mathcal{F}$ fading channels is given by
\begin{align}
	&P_\mathrm{out}=1-\notag\\
	&\left[1-\left[\sum_{n=1}^N \left(I_{\frac{m_1\gamma_\mathrm{th}}{m_1\gamma_\mathrm{th}+m_2}}^{\frac{4\eta_{n}^{(k)}}{2\eta_{n}^{(k)}-3}}\left(\frac{m_1}{2},\frac{m_2}{2}\right)-1\right)+1\right]^{\frac{2\eta_{n}^{(k)}-3}{4\eta_{n}^{(k)}}}\right]^K.\label{eq-out}
\end{align}
\begin{proof}
By inserting the instantaneous capacity region of the considered  dirty FA system from \eqref{eq-region} into the OP definition, we have
\begin{align}
&P_\mathrm{out}=\Pr\left(\frac{1}{2}\log_2\left(1+\underset{k=1,\dots,K}{\min}\frac{Ph_{\mathrm{FAS},k}^2}{\sigma^2}\right)\leq R_\mathrm{th}\right)\\
&=\Pr\left(\underset{k=1,\dots,K}{\min}h_{\mathrm{FAS},k}^2\leq \gamma_\mathrm{th}\right)\overset{(b)}{=}1-\left[1-F_{h_{\mathrm{FAS},k}^2}\left(\gamma_\mathrm{th}\right)\right]^K,\label{eq-p2}
\end{align}
where $\gamma_\mathrm{th}=\frac{2^{2R_\mathrm{th}}-1}{\bar{\gamma}}$ and $(b)$ is derived by assuming the independence of the fading channels $h_{n,k}$ for $k\in\{1,\dots,K\}$. Now, by inserting \eqref{eq-cdf3} into \eqref{eq-p2}, the OP is achieved as \eqref{eq-out} and the proof is completed.  
\end{proof}
\end{Proposition}

%
\begin{figure*}
\subfigure[$K=4$]{%
\includegraphics[width=0.35\textwidth]{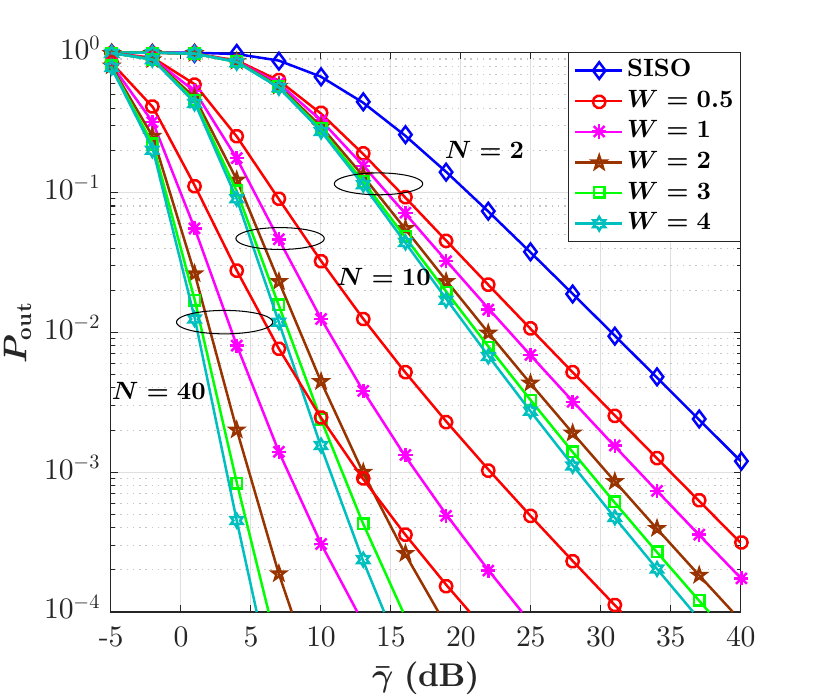}\label{outg4}%
}\hspace{-0.3cm}
\subfigure[$K=16$]{%
\includegraphics[width=0.35\textwidth]{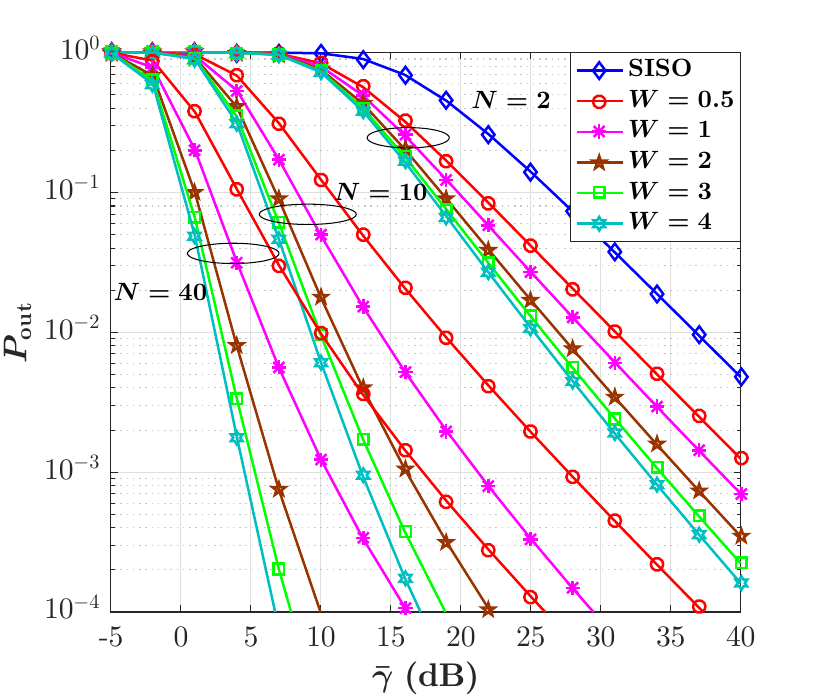}\label{outg16}%
}\hspace{-0.3cm}
\subfigure[$K=32$]{%
\includegraphics[width=0.35\textwidth]{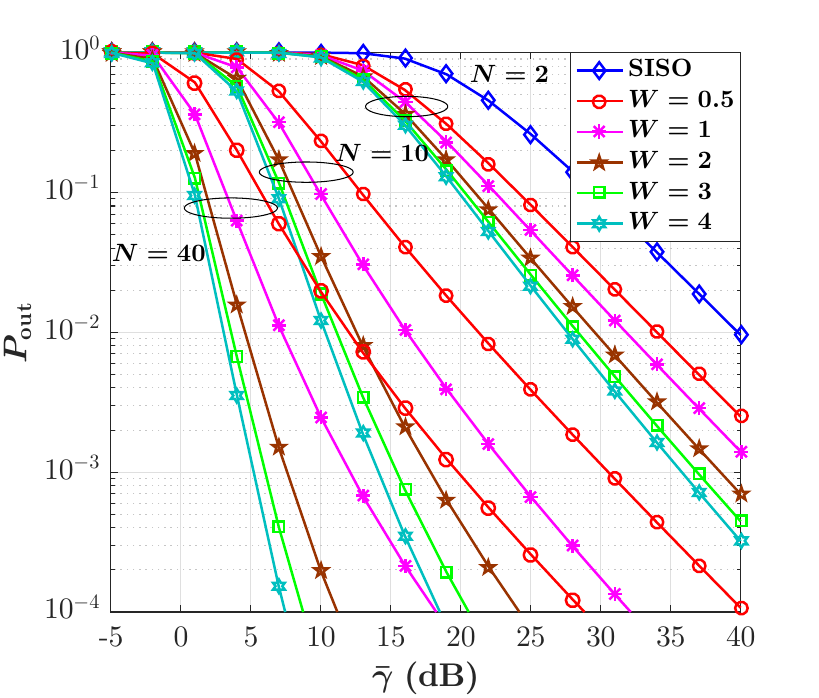}\label{outg32}%
}
\caption{OP versus average SNR for selected values of the FA size $W$, number for ports $N$, and number of users $K$.}\label{fig-outg}\vspace{-0.3cm}
\end{figure*}

\begin{figure*}
\subfigure[$K=4$]{%
\includegraphics[width=0.35\textwidth]{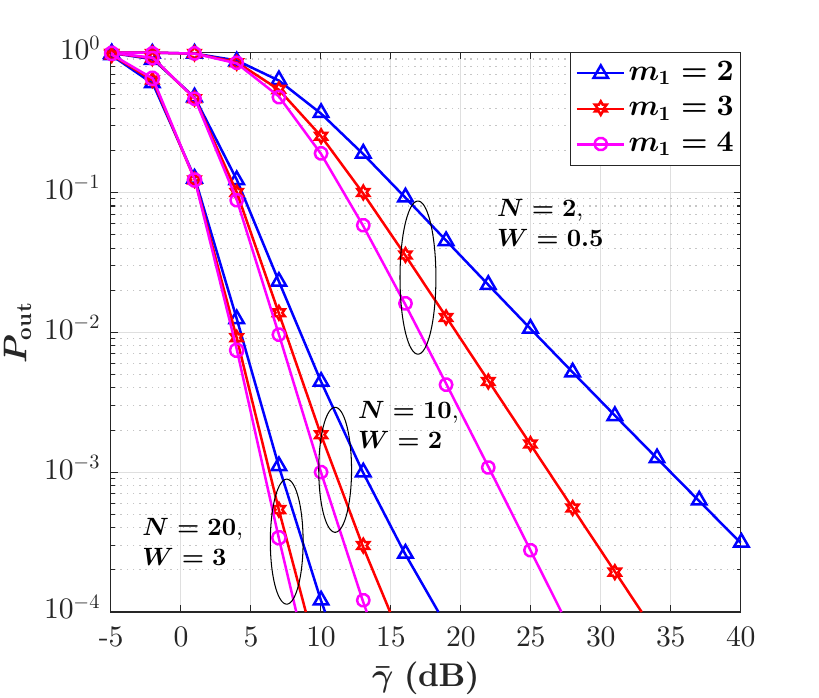}\label{out_m4}%
}\hspace{-0.3cm}
\subfigure[$K=16$]{%
\includegraphics[width=0.35\textwidth]{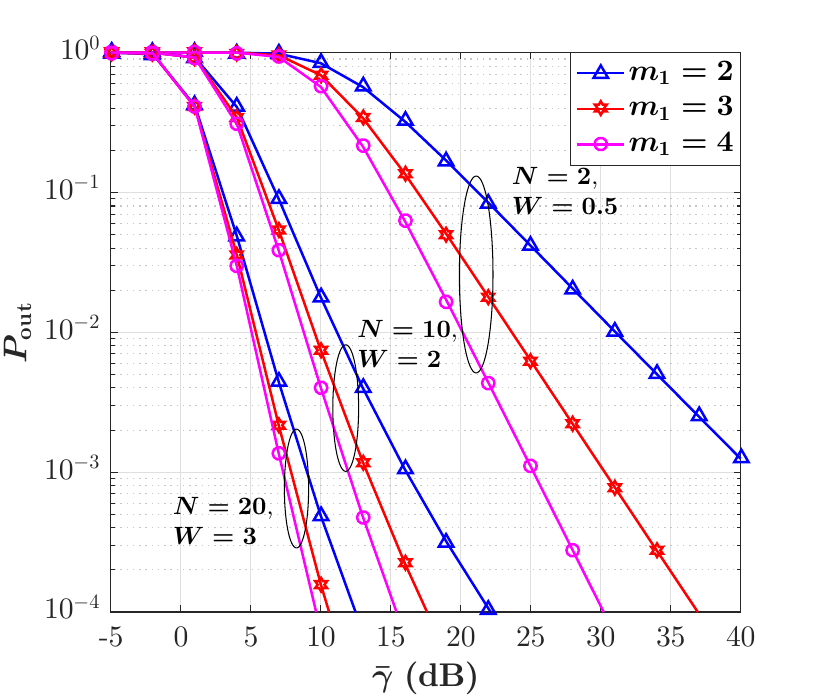}\label{out_m16}%
}\hspace{-0.3cm}
\subfigure[$K=32$]{%
\includegraphics[width=0.35\textwidth]{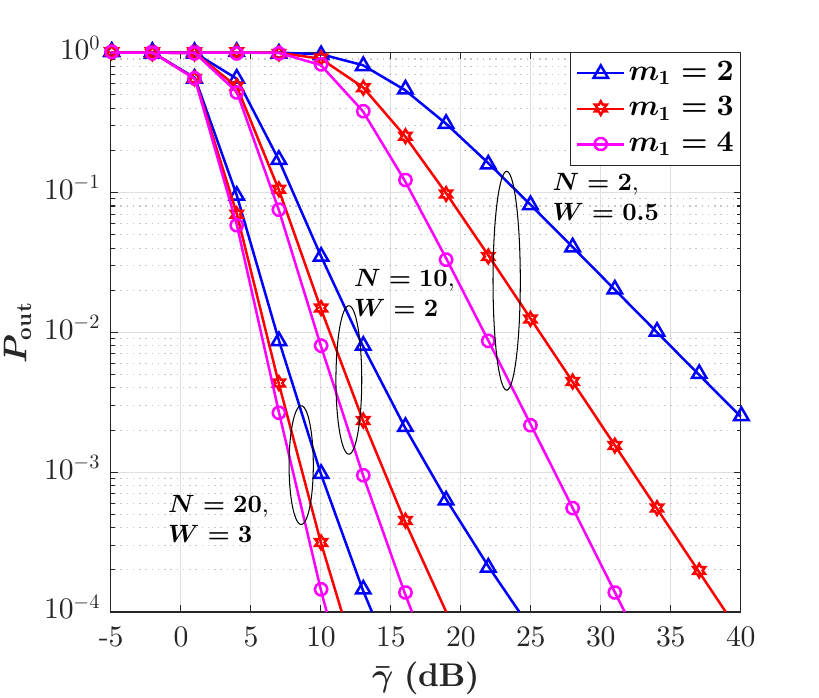}\label{out_m32}%
}
\caption{OP versus average SNR for selected values of the degree of freedom $m_1$, FA size $W$, number for ports $N$, and number of users $K$.}\vspace{-0.2cm}\label{fig-outm}
\end{figure*}
\vspace{-0.3cm}

\section{Numerical Results}\label{sec-num}
Here, we provide numerical results to evaluate the proposed analytical expression. Particularly, we analyze the behavior of the OP over the considered multiuser FA system in different scenarios when the number of users is $K=4, 16, 32$. We also set the simulation parameters as $P=30$dBm, $\sigma^2=-80$dBm, $R_\mathrm{th}=1$bps/Hz, $\delta^2=1$, and $(m_1,m_2)=(2,4)$. Additionally, to describe the dependence structure between the FA ports and generate the corresponding correlated Fisher-Snedecor $\mathcal{F}$ RVs (i.e., $h_{n,k}$), we simulate the Clayton copula by exploiting the Marshall and Olkin’s approach \cite{melchiori2006tools}, as shown in Algorithm \ref{al-1}, which is based on the Laplace transform.

Fig.~\ref{fig-outg} illustrates the performance of the OP in terms of average SNR $\bar{\gamma}$ for given values of the FA size $W$, the number of FA ports $N$, and the number of users $K$ under correlated Fisher-Snedecor $\mathcal{F}$ fading. For $K=4,16,32$, it can be seen that by increasing $\bar{\gamma}$, the OP decreases since the channel conditions between the transmitters and the common receiver improve. We can also see that considering the FA in multiuser communication systems can significantly lower the OP compared with the single-input single-output (SISO) system regardless the number of users. Particularly, we observe that for a given number of FA ports (e.g., $N=2,10,40$), the OP reduces as $W$ increases. The reason is that by increasing the spatial separation between the FA ports (i.e., increasing the FA size) for a fixed $N$, the spatial correlation between the corresponding channels becomes weaker; hence, the system performance in terms of the OP is ameliorated. In addition, we observe that for a fixed FA size $W$, the OP performance improves as the number of FA ports raises. However, such an improvement is much more noticeable for the large FA size (e.g., $W=4$) compared to the small values (e.g., $W=0.5$). It can also be seen that by simultaneously increasing $W$ and $N$, the best performance for the OP is achieved. Moreover, by comparing the results in Figs.~\ref{outg4}, \ref{outg16}, and \ref{outg32} from the number of users perspective, we can observe that as $K$ grows, the OP performance deteriorates since the network interference increases. Moreover, it can be found that one FA at the common receiver can significantly enhance the system performance in the presence of a large number of users. 

Given the importance of accurate fading channel modeling on the performance in real propagation environments, we study the impact of fading severity on the considered multiuser FA system in Fig.~\ref{fig-outm} by changing the degree of freedom $m_1$ which mainly refers to as the fading parameter in Fisher-Snedecor $\mathcal{F}$ distribution. For all selected values of $K$, $N$, and $W$, it can be seen that as the fading severity decreases (i.e., $m_1$ increases), the OP performance is improved, meaning that a lower value of the OP is achieved under a mild fading condition (e.g., $m_1=4$) than when a stronger one (e.g., $m_1=2$) is considered. In addition, we can observe that the impact of fading parameter $m_1$ is more tangible on the OP performance when the FA size and the number of FA ports are large. Moreover, it can be found from the results in Figs.~\ref{out_m4}--\ref{out_m32} that as the number of users $K$ grows, the effect of fading parameter $m_1$ becomes minor for a fixed average SNR $\bar{\gamma}$.

\begin{algorithm}
\caption{Clayton copula sampling}\label{al-1}
\textbf{Step 1.} \textit{Generate $J\sim\mathcal{L}^{-1}\left[\phi(t)\right]$, such that $J\sim\Gamma\left(\frac{1}{\beta},\beta\right)$ }\\
\textbf{Step 2.} \textit{Generate $\vec{S}=\left(S_1,\dots,S_d\right)$,  such that $S_i\sim\mathcal{U}\left(0,1\right)$ for $i=1,\dots,d$}\\
\textbf{Step 3.} \textit{Return $U_i=\phi\left(-\frac{\log\left(S_i\right)}{J}\right)$, such that $\phi(t)=\frac{t^{-\beta}-1}{\beta}$}
\end{algorithm}

\vspace{-0.3cm}
\section{Conclusion}\label{sec-con}
In this letter, we studied the performance of the FA-aided DMAC with non-causally known SI at the transmitters, where all fading channels follow Fisher-Snedecor $\mathcal{F}$ distribution. By integrating the copula theory with Jakes' model to describe the inherent spatial correlation among the FA ports, we derived an analytical expression for the joint CDF of the FA multiuser system. Then, we evaluated the considered system performance by obtaining the closed-form expression of the OP. Eventually, our results indicated that considering FA at the receiver can remarkably enhance the system performance in terms of the OP over multiuser communication systems.

\bibliographystyle{IEEEtran}
\bibliography{sample.bib}
\end{document}